\documentclass[envcountsame]{svproc}
\usepackage{amssymb} 
\usepackage{amsmath}
\usepackage{dsfont}
\usepackage{xspace}
\usepackage[numbers,sort&compress]{natbib}

\usepackage[margin=2cm]{geometry}

\usepackage{thm-restate}

\usepackage{xcolor}
\definecolor{REFCOLOR}{HTML}{B1063A}
\definecolor{HPIyellow}{rgb}{0.965, 0.659, 0}
\definecolor{HPIblue}{rgb}{0, 0.478, 0.62}
\definecolor{HPIorange}{rgb}{0.867, 0.38, 0.031}
\definecolor{HPIred}{rgb}{0.694, 0.024, 0.227}
\definecolor{HPIgreen}{rgb}{0, 0.698, 0.2}

\usepackage{hyperref}
\hypersetup{colorlinks=true, linkcolor=HPIblue, citecolor=REFCOLOR}
\usepackage{cleveref}

\usepackage{algorithm}
\usepackage{algpseudocode}

\usepackage{listings}

\usepackage{svg}
\usepackage{todonotes}
\usepackage{lineno}

\newcommand{\Sets}{\mathcal{S}}
\newcommand{\RLP}{Relaxed LP}
\newcommand{\IPS}{IP-Steiner}
\newcommand{\C}{C}
\newcommand{\iTrees}{\mathcal{T}_i}

\newcommand{\Expected}{\mathds{E}}

\newcommand{\bigO}{\mathcal{O}}
\newcommand{\ProblemName}[1]{\textsf{#1}}
\newcommand{\steinernumber}[1]{\sigma(#1)}
\newcommand{\depth}{\emph{depth}\xspace}

\usepackage[createShortEnv, commandRef=Cref]{proof-at-the-end}
\crefname{claim}{claim}{claims}
\Crefname{claim}{Claim}{Claims}
\crefname{lemma}{lemma}{lemmas}
\Crefname{lemma}{Lemma}{Lemmas}

\usepackage{url}

\begin{document}
\mainmatter              
\title{Optimal Approximations for the Requirement Cut Problem on Sparse Graph Classes}
\titlerunning{Optimal Approximations for the Requirement Cut Problem on Sparse Graph Classes}  
%
\author{Nadym Mallek \inst{1} \and Kirill Simonov\inst{2}}
\authorrunning{Nadym Mallek and Kirill Simonov} 
\institute{Hasso Plattner Institute, University of Potsdam, Potsdam, Germany\\
\email{nadym.mallek@hpi.de},\\
\and
Department of Informatics, University of Bergen, Norway\\
\email{k.simonov@uib.no}}

\maketitle              

\begin{abstract}
We study the \ProblemName{Requirement Cut} problem, a generalization of numerous classical graph partitioning problems including \ProblemName{Multicut}, \ProblemName{Multiway Cut}, \ProblemName{$k$-Cut}, and \ProblemName{Steiner Multicut} among others. Given a graph with edge costs, terminal groups $S_1, \ldots, S_g$ and integer requirements $r_1, \ldots, r_g$; the goal is to compute a minimum-cost edge cut that separates each group $S_i$ into at least $r_i$ connected components. Despite many efforts, the best known approximation for \ProblemName{Requirement Cut} yields a double-logarithmic $\bigO(\log(g)\cdot\log(n))$ approximation ratio as it relies on embedding general graphs into trees and solving the tree instance. 

In this paper, we explore two largely unstudied structural parameters in order to obtain single-logarithmic approximation ratios: (1) the number of minimal Steiner trees in the instance, which in particular is upper-bounded by the number of spanning trees of the graphs multiplied by $g$, and (2) the \depth of series-parallel graphs. Specifically, we show that if the number of minimal Steiner trees is polynomial in $n$, then a simple LP-rounding algorithm yields an $\bigO(\log n)$-approximation, and if the graph is series-parallel with a constant depth then a refined analysis of a known probabilistic embedding yields a $\bigO(\depth \cdot \log g)$-approximation on series-parallel graphs of bounded depth. Both results extend the known class of graphs that have a single-logarithmic approximation ratio.
\end{abstract}

\section{Introduction}
\label{sec:intro}

Graph partitioning problems are a central topic in combinatorial optimization, with applications such as from network design \cite{donde_identification_2005}, VLSI Layout \cite{cong_multilevel_2003}, image segmentation \cite{grady_isoperimetric_2006, boykov_fast_2001}, parallel computing \cite{bader_graph_2013} and even bioinformatics \cite{junker_analysis_2008}. In this paper, we study the \ProblemName{Requirement Cut} problem, a generalization of several well-known cut problems. Given an undirected graph $G = (V, E)$ with edge costs, a collection of $g$ terminal groups $ \Sets = (S_1, \dots, S_g) \subseteq V$, and a requirement $r_i > 1$ for each group, the goal is to compute a minimum-cost edge cut $C \subseteq E$ that separates each terminal group $S_i$ into at least $r_i$ components in $G'=(V,E\setminus C).$ 

The \ProblemName{Requirement Cut} problem is a natural extension of multiple classical problems that have been studied for decades. In the following, we present those problems, reference the paper that introduced it and summarize the current best general and parameterized approximation ratios: 
\begin{itemize}
    \item \ProblemName{Multicut} \cite{hu_multi-commodity_1963} : Given a set of terminal pairs $(s_1, t_1), \dots, (s_k, t_k)$, the goal is to find a minimum-cost edge cut that separates each pair. This corresponds to the \ProblemName{Requirement Cut} problem when each terminal group consists of exactly two vertices and has a requirement of $r_i = 2$. Garg et al. \cite{garg_approximate_1993} gives the best known approximation for general graphs with a $\bigO(\log(k))$-approximation ratio via LP rounding (region growing). Friedrich et al. give a constant approximation for series-parallel graphs \cite{friedrich_primal-dual_2022} and a more general  \cite{friedrich_approximate_2023} $\bigO(\log(tw))$ for graphs of treewidth $tw$.
    
    \item  \ProblemName{Multiway Cut} \cite{dahlhaus_complexity_1994}: Given a set of terminals $S = \{s_1, \dots, s_k\}$, the objective is to find a minimum-cost edge cut that separates all terminals into distinct components. This is a special case of \ProblemName{Requirement Cut} where there is only one group and $r = k$. Chopra et al. \cite{chopra_multiway_1991} studied the linear program of this problem in depth and got a $2-\frac{2}{k}$ approximation ratio. After several improvements and new techniques used \cite{calinescu_improved_2000, karger_rounding_2004, buchbinder_simplex_2018}, this chain culmianted with a 1.2965 approximation by Sharma and Vondràk \cite{sharma_multiway_2014} . This problem also has seen a relatively new variant in the past few years called the \ProblemName{Norm Multiway Cut} \cite{chandrasekaran_l_p-norm_2021, carlson_approximation_2023} which is a variant where the goal is to minimize the norm of the edges in the boundaries of $k$-parts of the graph.
    
    \item \ProblemName{Steiner-Multicut} \cite{klein_approximation_1997}: Given sets of terminals $\Sets = (S_1, \dots, S_k)$ each of size $(t_1, \dots, t_k)$, the goal is to find a minimum-cost edge cut that separates each set in at least 2 components. This corresponds to the \ProblemName{Requirement Cut} problem when each requirement is 2. This problem has particularly many applications has even been studied in a parametrized complexity setting \cite{bringmann_parameterized_2016}. Its best current approximation comes directly from the \ProblemName{Requirement Cut} and comes with a $\bigO(\log(k) \log(n))$ ratio. Moreover, it is a direct generalization of a less studied variant when $k=1$, the \ProblemName{Steiner-Cut} problem \cite{jue_near-linear_2019}. 
    
    \item \ProblemName{Multi-Multiway Cut} \cite{avidor_multi-multiway_2007}: A generalization of \ProblemName{Multiway Cut} where multiple terminal groups must be separated into connected components. This corresponds to the \ProblemName{Requirement Cut} problem when each terminal group has a requirement equal to its size. The best known approximation comes from the paper that introduced the problem \cite{avidor_multi-multiway_2007} and is $\log(2k)$. Slightly later two more papers have studied the problem under different angles, \cite{shuguang_multi-multiway_2009} by adding label edges and \cite{deng_multi-multiway_2013} studying the problem on bounded branch-width graphs.
    
    \item \ProblemName{k-Cut} \cite{goldschmidt_polynomial_1994} : The problem of partitioning a graph into at least $k$ connected components using a minimum-cost edge cut. This can be viewed as a \ProblemName{Requirement Cut} instance where there is only one terminal set covering all the vertices and its requirement $r = k$ ($k$ parameter given by the problem). First, the \ProblemName{3-cut} was introduced and solved in polynomial time by Hochbaum and Schmoys \cite{hochbaum_ov_1985}. It was then extended to \ProblemName{k-cut} problem by Goldsmith and Hochbaum \cite{goldschmidt_polynomial_1994} who gave a polynomial time algorithm for any fixed $k$. Later, by better understanding its parametrized complexity, Saran and Vazirani\cite{saran_finding_1995} found the current best approximation ratio $2-\frac{2}{k}$. More recently, Gupta et al. \cite{gupta_optimal_2021}, gave tight bounds for the run-time of computing an optimal solution.
    
    \item \ProblemName{Steiner k-Cut} \cite{chekuri_steiner_2006}: This is a generalization of the \ProblemName{k-Cut} problem, where $S_1 \subseteq V, |S_1| \geq k$ and the objective is to find a minimum cost set of edges whose removal results in at least $r_1 = k$ disconnected components, each containing a terminal. When $k =n$ this is exactly the \ProblemName{k-Cut} problem. The problem has been established and approximated best by Chekuri et al. \cite{chekuri_steiner_2006} to a $2-\frac{2}{k}$ ratio. It is worth noting that their work has established the Linear Program formulation used to solve the \ProblemName{Requirement Cut} Problem introduced later.
\end{itemize}

\begin{figure}[ht]
\small
	\begin{center}
		\begin{tikzpicture}
	 \node at (0,0){\includegraphics[scale = 0.8]{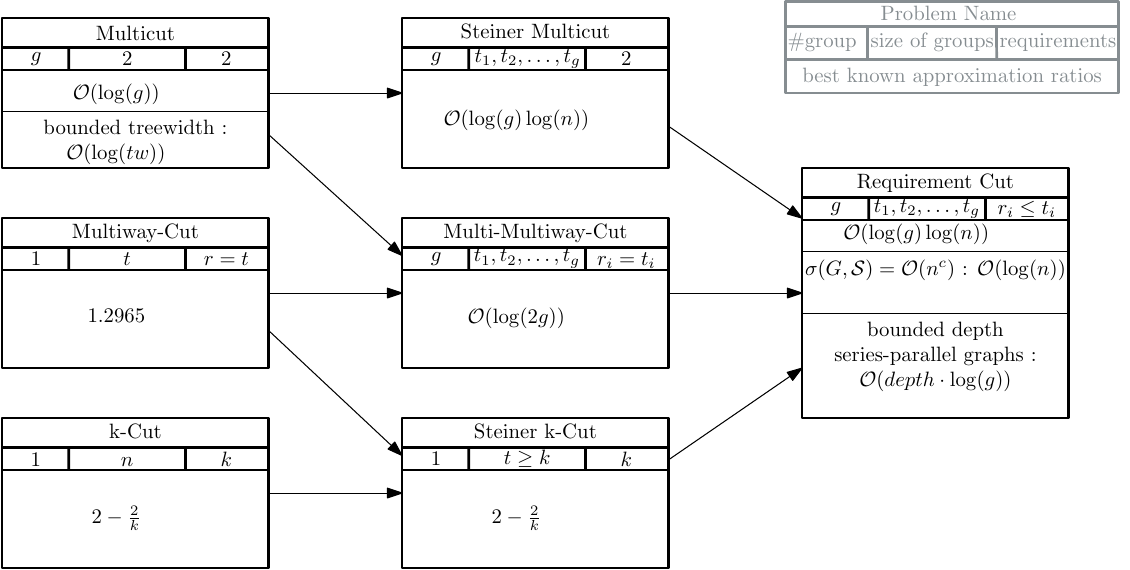}};				
            \node at (-5.2,2.57){\cite{garg_approximate_1993}};
            \node at (-5.1,1.78){\cite{friedrich_approximate_2023}};
            \node at (-5.35,-0.45){\cite{sharma_multiway_2014}};
            \node at (-5.42,-3.15){\cite{saran_finding_1995}};
            \node at (0.6,2.23){\cite{gupta_improved_2010}};
            \node at (0.2,-0.45){\cite{avidor_multi-multiway_2007}};
            \node at (-0.1,-3.15){\cite{chekuri_steiner_2006}};
            \node at (6,0.68){\cite{gupta_improved_2010}};
            \node at (5.1,-0.15){(\Cref{sec:GraphsBST})};
            \node at (5.1,-1.6){(\Cref{sec:depthspg})};
		\end{tikzpicture}   
	\end{center}
	\caption{Hierarchy of graph cut problems captured by the Requirement Cut framework putting our results in perspective. Boxes contain parameters, and best known approximation results for general graphs (and for more particular cases when specified). Arrows denote special cases of the problems (i.e., the source is a particular case of the target).}
	\label{fig:cutproblems}
\end{figure}

These problems have been studied mostly through the lens of approximation algorithms. Some of them have received particular attention on specific classes of graphs with parameters such as treewidth, branchwidth or pathwidth.

The hardness of the \ProblemName{Requirement Cut} can be shown through any of these problems since they are all hard and are particular cases of it. A particular reduction that helps us with bounding the approximability on trees is a reduction from \ProblemName{Set Cover} taken from \cite{nagarajan_approximation_2010}. Consider a reduction from \ProblemName{Set Cover} to the \ProblemName{Requirement Cut} problem on trees. We construct a star graph where each leaf node corresponds to a set in the \ProblemName{Set Cover} instance, and the center node of the star connects to each of these leaves via an edge of cost $1$. For each element $j$ in the universe, we define a group consisting of the root and all leaf nodes whose corresponding sets contain element $j$. We assign a requirement of $2$ to each group.
In this construction, selecting an edge in the tree (i.e., choosing a set in the Set Cover instance) ensures that the corresponding group is cut (i.e., its requirement is satisfied) if and only if that set covers one of the group's elements. Hence, there is a one-to-one correspondence between feasible solutions of the \ProblemName{Requirement Cut} instance and valid set covers, which cost the same. Feige et al. \cite{feige_threshold_1998} showed that \ProblemName{Set Cover} can not be approximated in polynomial time to a $\Omega(\log(n))$ ratio where $n$ is the size of the universe to cover. It follows that \ProblemName{Requirement Cut} on trees is hard to $\Omega(\log(g))$- approximate where $g$ is the number of groups.

Now that we have seen that the best approximation possible on trees for our problem is $\log(g)$, let us put the results of our paper in context. Note that, the reduction from \ProblemName{Set Cover} to \ProblemName{Requirement Cut} on trees works even when all the requirements are 2 (i.e. the \ProblemName{Requirement Cut} instance is a \ProblemName{Steiner Multicut} instance). In 1997, Klein, Plotkin, Rao and Tardos \cite{klein_approximation_1997} presented an algorithm that yields $\bigO(\log^3(gt))$-approximation with $t = \max(t_i)$ (the sizes of the sets) for the \ProblemName{Steiner Multicut} problem. This ``triple log'' approximation ratio remained the best for this problem until 2010. In 2010, the \ProblemName{Requirement Cut} problem was introduced by Nagarajan and Ravi \cite{nagarajan_approximation_2010} who proved a $\bigO(\log(n) \cdot \log(gR))$-approximation with $R = \max(r_i)$. This approximation became, at the same time, the best approximation ratio for the \ProblemName{Steiner Multicut} problem. Their bound was later improved by the same set of authors along with Gupta \cite{gupta_improved_2010} that got the approximation ratio down to $\bigO(\log(n) \cdot \log(g))$, effectively removing all dependency on the requirements. Both involve very similar algorithmic steps:
\begin{enumerate}
    \item Finding a good approximation algorithm for trees, $\log(gR)$ for \cite{nagarajan_approximation_2010} and $\log(g)$ for \cite{gupta_improved_2010}. 
    \item Using Fakcharoenphol, Rao and Talwar (FRT) \cite{fakcharoenphol_tight_2004}, they embed the graph into a tree and pay a $\log(n)$ factor for it.
    \item They use the the first approximation on the resulting tree and get their $\bigO(\log(gR) \cdot \log(n))$ or $\bigO(\log(g) \cdot \log(n))$-approximation.
\end{enumerate}
It is also worth noting that Gupta et al. \cite{gupta_improved_2010} technically has a slightly better $\bigO(\log(k) \cdot \log(g))$-approximation with $k$ being the total number of terminals. They obtain this improvement through something that they call a slight generalization of FRT \cite{fakcharoenphol_tight_2004} in which they embed into a tree only the terminal vertices and hence only pay $\log(k)$ and not $\log(n)$. But this remains secondary in their paper since their main contribution was closing the gap between the $\log(gR)$ approximation on trees and the known lower bound of $\log(g)$ making it tight.

For the past few years, one goal of the community has been to beat those ``double log'' approximations with ``single log'' ones. Since we know that embedding into trees has a $\log(n)$ cost, (surprisingly) even for treewidth 2 graphs \cite{gupta_cuts_2004}, and that a \ProblemName{Requirement Cut} problem instance on trees can not be approximated better than with a $\log(g)$ factor \cite{feige_threshold_1998}, we know that continuing this path of using tree embeddings will not result into a ``single log'' approximation. In this paper, we try something different over two different routes that we present in the next section.  
\subsection{Our Contributions}

We address the following question: \emph{For which classes of graphs does the \ProblemName{Requirement Cut} problem admit optimal (single-log) approximation algorithms?} To this end, we identify two structural parameters—(1) the \textbf{number of spanning or minimal Steiner trees} of the input graph and (2) the \textbf{depth} of series-parallel graphs. We show that both these relatively unstudied parameters offer better approximation for the  \ProblemName{Requirement Cut} problem. 

First, we demonstrate that graphs with a polynomial number of spanning trees (or, more generally, \textbf{number of minimal Steiner trees}) admit an $\bigO{\log (n)}$-approximation via a randomized LP-rounding scheme similar to the one given by Nagarajan and Ravi  \cite{nagarajan_approximation_2010}.

Second,  we show that for series-parallel graphs of constant \depth, a refined analysis of Emek and Peleg's \cite{emek_tight_2010} existing embeddings method yields an $\bigO{\log(g)}$-approximation, extending the known $\bigO{log(g)}$-approximations for trees to a broader class of graphs. 

The results broaden the landscape of instances that can be approximated with a ``single logarithmic'' factor and shed light on two unstudied parameters until now, $depth$ of series-parallel graphs and the order of magnitude of the number of spanning trees in a graph.

\begin{itemize}
    
    \item \textbf{Instances with a Polynomial Number of Steiner Trees}: A key parameter in our analysis of \Cref{alg:RRRC} of \Cref{sec:MainAlgorithm} depends on the instance. We refer to it as $\steinernumber{G,\Sets}$ for a given \ProblemName{Requirement Cut} instance. This function is equal to the number of distinct \hyperref[def:minsteiner]{Minimal Steiner Trees} (Steiner Trees for which all the leaves are terminal). This quantity is hard to compute \cite{valiant_complexity_1979} and so we will almost always upper-bound it by $\tau(G)\cdot g$ with $\tau(G)$ being the number of distinct spanning trees of the graph which is computable in polynomial time using Kirchhoff's Matrix Tree Theorem initially stated in \cite{kirchhoff_ueber_1847}. We show that if $\tau(G) = \bigO(n^{c_1})$ then necessarily $\steinernumber{G,\Sets} = \bigO(n^{c_2})$ when $c_1$ and $c_2$ are constants. 

    \item \textbf{Bounded-Depth Series-Parallel Graphs}: Series-parallel graphs are well known for their recursive structure. They are often studied as a generalization of trees since they are, for example, representative of treewidth 2 graphs (trees have treewidth 1) and also are exactly $K_4$-minor-free graphs (trees are $K_3$-minor-free). In our paper,  we introduce their \textbf{depth} as a key parameter. This parameter has never been studied to the best of our knowledge and we show a $\bigO(depth\cdot \log(g))$-approximation for the \ProblemName{Requirement Cut} problem on series-parallel graphs. Extending the class of graphs that have a $\bigO(\log(g))$-approximation to the series-parallel graphs with any constant $depth$. This class of graph can additionally be seen as a generalization of \textit{melonic graphs} which are used in other fields such as physics.

\end{itemize}
In \Cref{sec:GraphsBST} we will give several classes of graphs and instances for which our results apply. Afterwards in \Cref{sec:MainAlgorithm} we give, prove and analyze our algorithm regarding graphs of polynomial number of Steiner trees. This comes with a proof about the Linear Program formulation of the \ProblemName{Requirement Cut} problem that is of its own interest (\Cref{lem:ILP}). Finally in \Cref{sec:depthspg}, we refine Emek and Peleg's \cite{emek_tight_2010} analysis to adapt it to the depth of series-parallel graphs. This approach provides an alternative to the known $\log(n)$ embeddings for series-parallel graphs. Due to space restrictions, some technical proofs are moved to the appendix; such statements are marked with a $\star$.

\section{Preliminaries}
\label{sec:prelim}
\begin{definition}[Cut]
Let $G = (V, E)$ be an undirected graph. A \emph{cut} in $G$ is a set of edges $C \subseteq E$ such that removing the edges in $C$ disconnects some specified structure in the graph. We say that a cut \emph{separates} a set of vertices $S \subseteq V$ if the vertices in $S$ lie in at least two distinct connected components of $G \setminus C$.
\end{definition}

\begin{definition}[Requirement Cut Instance]
\label{def:RCI}
A \emph{Requirement Cut} instance is defined by an undirected graph $G = (V, E)$, a cost function $c: E \rightarrow \mathbb{R}_{\geq 0}$, a collection of terminal sets $\Sets = \{S_1, \ldots, S_g\}$ with $S_i \subseteq V$, and integer requirements $ r_1, \ldots, r_g \geq 2$. The goal is to find a minimum-cost edge set $C \subseteq E$ such that in the graph $G' = (V, E \setminus C)$, each terminal set $S_i$ is partitioned into at least $r_i$ connected components.
\end{definition}
\begin{definition}[Steiner Tree]
Let $G = (V,E)$ and $\Sets = \{S_1, \ldots, S_g\} \subseteq V$ a set of terminal sets. A \emph{Steiner tree} for a terminal set $S$ is a minimally connected subgraph $T \subseteq G$ that spans all vertices in $S$ and is called an $S$-Steiner tree. $T$ may include additional non-terminal vertices, called \emph{Steiner vertices}. An edge connecting two terminals (resp. Steiner) vertices may be called \emph{terminal edge} (resp. \emph{Steiner edge}). 
\end{definition}

Additionnaly we define the length of a Steiner tree $T$ as  $\sum_{e \in T} d_e$ for a given length function $d$ and $d_e$ the length of edge $e$. 

\begin{definition}[Minimal Steiner Tree]
\label{def:minsteiner}
A Steiner tree is said to be minimal if all its leaves are terminal vertices. The number of distinct minimal Steiner trees in a graph $G = (V,E)$ with terminal sets $\Sets = \{S_1, \ldots, S_g\}$, will be noted as $\steinernumber{G,\Sets}$. 
\end{definition}

With this definition of Minimal Steiner tree, let us prove the following claim:

\begin{theoremEnd}[normal]{claim}[]
\label{clm:spann-stein}
In a graph \(G = (V, E)\) with \(\Sets = \{S_1, \ldots, S_g\}\) terminal sets, we have $\steinernumber{G,\Sets} \leq g \cdot \tau(G)$. With $\tau(G)$ the number of distinct spanning trees in $G$.
\end{theoremEnd}
\begin{proofEnd} 
Let \(T\) be a minimal Steiner tree of set \(S_k\) for some \(1 \leq k \leq g\). By adding edges and vertices to \(T\) until it covers \(V(G)\), one can create at least one spanning tree per minimal Steiner tree. Also, two distinct minimal Steiner trees for the same terminal set cannot extend to the same spanning tree in \(G\), as they differ in at least one edge. Therefore, there exist at most \(\tau(G)\) minimal Steiner trees for each terminal set \(S_i \in \Sets\).
\end{proofEnd}

\begin{definition}[Series-Parallel Graphs]
\label{def:spGraphs}
A \emph{series--parallel graph} (SP graph) is any graph that can be obtained from a single edge $K_2$ by a finite sequence of the following operations:
\begin{itemize}
    \item \textbf{Parallel composition:} Replace an edge $(u,v)$ into a set of two parallel edges that connect $u$ and $v$.
    \item \textbf{Series composition:} Replace an $(u,v)$ into a degree 2 vertex connected to $u$ and $v$.
\end{itemize}

\end{definition}

\begin{definition}[Composition Trace]
\label{def:comptrace}
Let $G$ be a series-parallel graph constructed via a sequence of series and parallel compositions. The \textbf{composition trace} of $G$ is a directed tree $T$ whose:
\begin{itemize}
    \item leaves correspond to the individual edges of $G$,
    \item internal nodes represent subgraphs formed by (series or parallel) composition of their children
    \item root corresponds to $G$ itself.
\end{itemize}
We group consecutive series or parallel compositions into single nodes, so that every path from a leaf to the root alternates between series and parallel composition steps. A node in $T$ may have arbitrarily many children.
\end{definition}

\begin{definition}[Depth of a Series-Parallel Graph]
\label{def:depth}
Let $G$ be a series-parallel graph and let $T$ be its composition trace. The \emph{depth} of $G$, is defined as the depth of the tree $T$, that is, if $l$ is the length of the longest path from a leaf (corresponding to a single edge) to the root (corresponding to $G$ itself): \depth $= l- 1$ .
\end{definition}

For example, a single edge would be a series-parallel graph of depth 0, a path or a many parallel edges would have depth 1 and a composition of parallel paths would have depth 2. This means that this definition makes each class of bounded depth series-parallel graphs an infinite class of graphs.

\section{Graphs with Polynomial Number of Steiner Trees}
\label{sec:GraphsBST}
In this section, we will see for which graphs our \Cref{alg:RRRC}, introduced later in \Cref{sec:MainAlgorithm}, is more efficient than the state of the art. First, we will comment on the link between the number of spanning trees and the number of Steiner trees, and then we will also show that the number of spanning trees is not really limiting our algorithm and that there are other graphs where our algorithm provides an efficient approximation as well. 
\begin{figure}[ht]
\small
	\begin{center}
		\begin{tikzpicture}
	 \node at (0,0){\includegraphics[scale = 0.8]{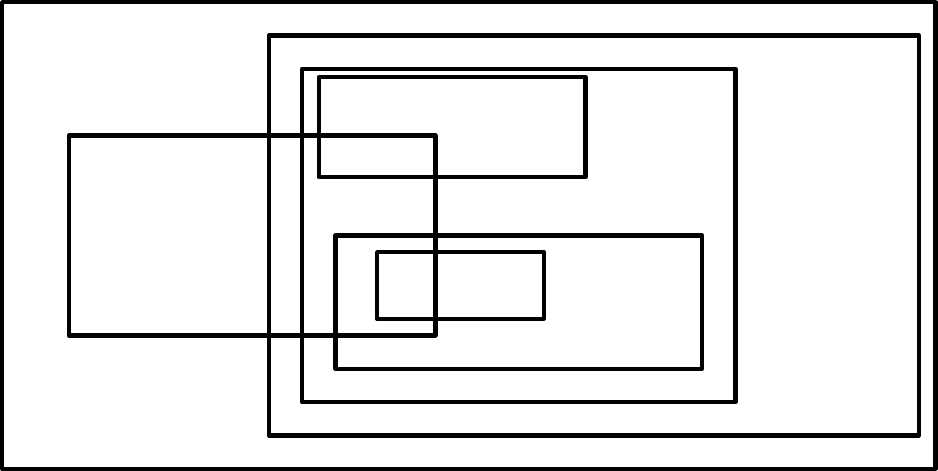}};				
            \node at (-4.1,2.9){All Requirement Cut Instances};
            \node at (-3.9,0.5){Bounded Depth};
            \node at (-3.9,0.1){Series-Parallel};
            \node at (-3.9,-0.3){(\Cref{sec:depthspg})};
            \node at (-0.1,1.9){Few Short Cycles};
            \node at (-0.1,1.5){(\Cref{sec:GraphsBST})};   
            \node at (0.4,-1.6){Bounded Feedback Edge Set};
            \node at (-0.,-0.7){Tree};
            \node at (2.3,0.5){$\tau(G) = \bigO(n^c)$};
            \node at (4.82,0.7){$\steinernumber{G,\Sets} = \bigO(n^c)$};
            \node at (4.82,0.3){(\Cref{sec:GraphsBST})};            
		\end{tikzpicture}  
	\end{center}
	\caption{Overview of graph classes and instances for which the \ProblemName{Requirement Cut} problem admits a single-logarithmic approximation in this paper.}
	\label{fig:instances-RC}
\end{figure}
\subsection{Graphs with Polynomial number of Spanning Trees} 

A key parameter in our analysis is the number of spanning trees in a graph, which we denote as $\tau(G)$. By Kirchhoff’s Matrix Tree Theorem \cite{kirchhoff_ueber_1847}, $\tau(G)$ can be computed as the determinant of any cofactor of the graph Laplacian matrix of the graph. Equivalently, if $\lambda_2, \lambda_3, \dots, \lambda_n$ are the nonzero eigenvalues of $L$, then $\tau(G) = \frac{1}{n} \prod_{i=2}^{n} \lambda_i$. In our setting, we focus on graphs where $\tau(G)$ is polynomial in $n$, i.e., $\tau(G) = \bigO(n^c)$ for some constant $c$. 

The number of spanning trees multiplied by $g$ (the number of terminal sets of our \ProblemName{Requirement Cut} problem instance) is a natural upper bound to the number of distinct minimal Steiner trees in that same instance. So naturally, graphs with a polynomial number of spanning trees satisfy our condition. First, let us now see two examples of graph classes that satisfy our spanning tree condition. Second we will also illustrate the kind of instances where the number of Steiner-tree is polynomial and for which any support graph, regardless of number of spanning trees, would admit a  ``single-log'' approximation.

\subsubsection{Graphs with Small Feedback Edge Sets} 

One prominent category of graphs with polynomially many spanning trees is the graphs that have a small, or bounded, feedback edge set. The feedback edge set is defined as the minimum number of edges whose removal results in an acyclic graph. Note that the feedback edge set number of a connected n-vertex and m-edge graph G is always $m - n + 1$ and a minimum feedback edge set can be determined by the computation of a spanning tree in $\bigO(n+m)$ time via depth-first search. Also note that this class of graphs contains graphs with bounded number of cycles.

Back to our case, if the size of a graph's feedback edge set is bounded by a constant $k$, then the number of spanning trees is $\binom{n-1+k}{k} = \bigO(n^k)$, and hence polynomially bounded in the size of the graph.

\subsubsection{Graphs with few Short Cycles}
Graphs with short cycles, meaning each cycle contains up to $k$ vertices and the number of cycles is of order $\bigO(\log(n))$, have a polynomial number of spanning trees. In such a graph, a spanning tree is obtained by removing $1$ edge for each cycle making the number of spanning trees in the order of $\bigO(k^{\log(n)}) = \bigO(n^{\log(k)})$ which is polynomial for fixed k. 

\begin{figure}[ht]

	\begin{center}
		\begin{tikzpicture}
	   \node at (0,0){\includegraphics[scale = 1]{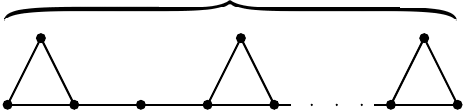}};		     \node at (0,1.5){$\log(n)$ cycles of bounded length};	           
		\end{tikzpicture}  
	\end{center}
	\caption{This graph has $\bigO(\log(n))$ cycles of length 3. each spanning tree chooses 1 edge per cycle to remove, hence yielding a number of spanning trees in $\bigO(3^{\log(n)})$ which is polynomial $\bigO(n^{\log(3)})$.}
	\label{fig:shortcycles}
\end{figure}


\subsection{Instances of Polynomial number of Steiner Trees}

The Steiner trees we are interested in for our algorithm are minimal Steiner trees, as given by~\Cref{def:minsteiner}. We denote the number of such trees for a given instance by $\steinernumber{G,\Sets}$. This value depends on the terminal sets of the instance as well as on the graph, and is hence more specific than its upper bound $\tau(G) \cdot g$. Unfortunately, it is hard to compute. Still, it can provide better approximation for some type of instances where the $\tau(G)$ would yield a value that is super-polynomial in $n$.

In those instances the whole graph may have a super-polynomial number of spanning trees but $\steinernumber{G,\Sets}$ is in $poly(n)$. This can occur when the terminal sets are fully located in subgraphs that each have a polynomial number of spanning trees and are connected to the rest of the graph by exactly one edge.

\begin{figure}[ht]

	\begin{center}
		\begin{tikzpicture}
	   \node at (0,0){\includegraphics[scale = 1]{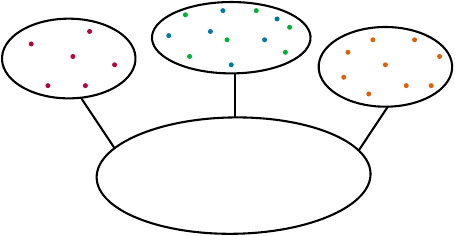}};		     
        \node at (0.1,-1){$G = (V,E)$};
        \node at (-2.8,2){$H_1 \supseteq (\textcolor{HPIred}{S_1})$};
        \node at (0,2.3){$H_2 \supseteq (\textcolor{HPIgreen}{S_2}, \textcolor{HPIblue}{S_3})$};
        \node at (2.6,2){$H_3 \supseteq (\textcolor{HPIorange}{S_4})$};
		\end{tikzpicture}  
	\end{center}
	\caption{Instance where terminal groups reside in minimally connected subgraphs. Here, $\tau(H_i) = poly(|H_i|), i\in\{1,2,3\}$ would suffice to have a $\log(n)$ approximation for the instance.}
	\label{fig:PolySteinerFig}
\end{figure}

\section{Requirement Cut Algorithm}
\label{sec:MainAlgorithm}
\subsection{Linear Program}
We work with essentially the same $LPs$ as \cite{nagarajan_approximation_2010} and \cite{gupta_improved_2010}. However, in \Cref{lem:ILP}, we will provide additional insights about these linear programs, compared to these previous works. This is central to our analysis as it allows us to guarantee that our rounded solution is indeed a valid solution to the \ProblemName{Requirement Cut} problem. Both these $LPs$ work on slightly modified graphs w.r.t $G = (V,E)$. The first $LP$, \IPS{}, calls for a complete graph, so we add 0-cost edges to the graph until it is a clique, we call this graph $G_V$. The second, \RLP{}, also calls for extra edges only between terminals of the same group. We shall call this modified graph, where $S_i$ is a clique with the addition of 0-cost edges, $G_{S_i} = (V, E_{S_i})$. 
\[
\begin{aligned}
(\text{IP-Steiner}) \quad \min \quad & \sum_{e \in E} c_e d_e \\
\text{s.t.} 
\quad & \sum_{e \in \iTrees} d_e \geq r_i - 1 \quad && \forall \iTrees: S_i\text{-Steiner tree in } G_V, \quad \forall i = 1 \dots g \\
& d_e \in \{0,1\} \quad && \forall e \in E
\end{aligned}
\]
\[
\begin{aligned}
(\text{Relaxed LP}) \quad \min \quad & \sum_{e \in E} c_e d_e \\
\text{s.t.} \quad & \sum_{e \in \iTrees} d_e \geq r_i - 1  \quad \forall \iTrees: \text{spanning tree in } G_{S_i}, \forall i = 1 \dots g \\
& d_{\{u,w\}} \leq d_{\{u,v\}} + d_{\{v,w\}} \quad  \forall u,v,w \in V \\
& 0 \leq d_{\{u,v\}} \leq 1 \quad  \forall u,v \in V
\end{aligned}
\]

\IPS{} is an exact formulation of the \ProblemName{Requirement Cut} problem. However, we can notice that it is not clear why with only the constraint on the number of cut edges per Steiner tree, the solution would be a solution to the \ProblemName{Requirement Cut} problem. In \Cref{lem:ILP} we precisely show the role of the completeness of the graph in \IPS{}. In fact, without the completeness of the graph in \IPS{} (or of the terminal sets in \RLP{}), Steiner trees could have a long $d$ length without satisfying the requirement. For example,  take a path with any number of Steiner vertices in the middle and an arbitrary number of terminal vertices on each side. This graph could be cut $r-1$ times satisfying the constraint of the $LP$ while having its terminals lying only in 2 instead of $r$ components.

The algorithm we present in this section is a randomized rounding of the \RLP{} solution which can be solved by an ellipsoid method, using a Minimum Spanning Tree (MST) algorithm as an oracle. It ensures that each and every Steiner tree is cut at least $r - 1$ times. This alone does not suffice for proving that it returns a valid solution. Actually in  \cite[Claim 2]{nagarajan_approximation_2010} Nagarajan and Ravi use a preprocessing of their graph allowing them to have a relation between the number of edges removed from the graph and the number of terminal components created. This same claim is used in Gupta, Nagarajan and Ravi \cite[Claim 3]{gupta_improved_2010}. Both their approaches can use this claim because they work on trees. We need something stronger since we are working with a more general class of graphs. Hence, before presenting our algorithm we will show an equivalence (under the condition 
of 0-cost edges added of our $LP$) between the two statements ``All terminals of a set $S_i$ are in $r_i$ components'' and ``All spanning trees in $G_{S_i}$ are cut at least $r_i - 1$ times.''. Also, just like the two papers mentioned, let $d^*$ be the solution to \RLP{}, we work with $d_{\{u,v\}} = \min(2\cdot d^*_{\{u,v\}}, 1)$ which is also a metric since $d^*$ is one.
\begin{claim}[\cite{nagarajan_approximation_2010}] 
For any group $S_i, i\in (1,\dots, g)$, the minimum Steiner tree on $S_i$ under metric $d = \min(2\cdot d^*, 1)$ has length at least $r_i - 1$.
\end{claim}

\begin{theoremEnd}[normal]{lemma}[]
\label{lem:ILP}
Let $G = (V, E)$ a graph, $\Sets = \{S_1, \ldots, S_g\}$ with $S_i \subseteq V$ a collection of terminal sets, and integer requirements $ r_1, \ldots, r_g \geq 2$. Let each $G_{S_i}$ be the graph $G$ where 0-cost edges are added between nonadjacent pairs inside $S_i$ making $S_i$ a clique. Let $d: E \rightarrow \{0,1\}$ be a metric such that for every $S_i$-Steiner tree $T$, $d(T) \geq r_i - 1$. Then it is equivalent to say that all spanning trees over $G_{S_i}$ are cut at least $r_i - 1$ times and to say that the vertices of $S_i$ lie in $r_i$ components.
\end{theoremEnd}
\begin{proofEnd} 
For the purpose of this proof we drop the indices by considering the graph $G_S = (V, E_S)$ with $S$ being a terminal set with requirement $r$. It is enough to show the property for one terminal set since more terminal sets only cut the graph in more components by cutting more edges. We call an edge $e$ with $d(e) = 1$ a cut-edge. We call a component containing a terminal vertex, a terminal component. A non-terminal component is a component that does not contain a terminal vertex. 

Let $T_0$ be an $S$-Steiner tree in $G$ such that $d(T_0) = r-1$ (without loss of generality since more cut edges only work in our favor) and $S$ is in less than $r$ terminal components of $T_0 = (K_1,\dots,K_r)$. Let $K^* \in T_0$ be a non-terminal component. Let $\delta(K^*)$ be the number of \textit{cut-edges} that have exactly one end-point in $K^*$. $K^*$ \textbf{is not} a leaf component since $T_0$ is a minimal Steiner tree and hence $K^*$ has at least 2 \textit{cut-edges} incident to it, $\delta(K^*) \geq 2$. Also, let $(t_1, \dots, t_{\delta(K^*)}) \subseteq T_0\setminus K^*$ be the terminal \textbf{subtrees} (subtrees can contain several components) adjacent to $K^*$ through cut-edges $(e_1, \dots, e_{\delta(K^*)})$. Let $x\in t_1$ and $y\in t_2$ two terminal vertices and $e^* = (x,y)$ the 0-cost edge in $G_S$ connecting them. We call it the shortcut-edge.
 
With all this, we can \textit{shortcut} $K^*$ and reconstruct a tree $T_1$ that was considered by the LP and contains $e^*$ (the shortcut edge). We remove from the cut $e_1$ and $e_2$ the edges connecting $K^*$ to $t_1$ and $t_2$, and we add into the cut $e^*$ (necessary member of the cut since it costs 0 and achieves a desired separation) and an arbitrary edge of $T_1$, for example the closest to $e_1$ or $e_2$, say $e$. Now, we have $T_1$ another $S$-Steiner tree with $T_1 = \{e^*,e\}\cup T_0\setminus K^*$. By our construction so far $d(T_1) = d(T_0) = r-1$. More generally for any tree $T_i$ this procedure allows us to remove a non-terminal component and prove the existence of $T_{i+1}$ with a new set of edge-cuts. 

We effectively removed $K^*$ by short cutting it and we now have 2 cases:
\begin{enumerate}
    \item The number of non-terminal components is reduced. If it is at 0, we are done since there are at least $r$ components.
    \item The number of non-terminal components is the same. 
\end{enumerate} 
Indeed, the number of non terminal components cannot increase and can at most create one. Case $(2)$ simply means that the new cut-edge $e$ cut a terminal component into a non-terminal part and a terminal part. However, in that case, we are sure to either get closer to decreasing the number of non-terminal components since we shrink the size of the terminal part of that terminal component or to actually decrease it that number.

In some cases, the short cutting may shortcut more than one non-terminal components, in those cases, we reallocate all the cut edges of shortcut components (cut edges that are not in $T_{i+1}$) the same way we did for $e$ by shifting them to the nearest not cut edge of $T_{i+1}$.

One last thing that remains to show is that we have enough shortcuts to shortcut all non-terminal components that may appear. Note that there at most $r-2$ non terminal components since at least $2$ components contain terminals. Also note that $r$ is bounded by the number of terminal vertices by definition of the \ProblemName{Requirement Cut} Problem. To shortcut those $r-2$ non terminal components, we have at least $r-1$ shortcuts since the worst case for number of shortcuts of available shortcut edges happens when many non-terminal components separate the same (up to a difference of other non-terminal components) terminal components. And the least number of shortcuts in that case happens when on one side of the non-terminal components there is exactly one terminal vertex and the rest on the other side. When this happens, we have $r-1$ shortcuts for $r-2$ non terminal components.

\begin{figure}[ht]

\begin{center}
\includegraphics[scale = 0.9]{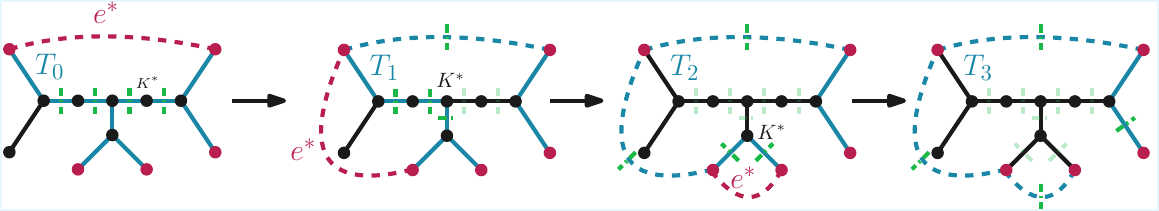}		    
\end{center}
\caption{In this figure, the successive considered trees are in blue. The cut-edges indicated in green and the shortcut used to shortcut $K^*$ is in red like the terminal vertices. Notice how from the first to the second stage, we did not make progress in number of non-terminal components but still made progress towards cutting the terminals in the component. }
\label{fig:lemma10fig}
\end{figure}
\end{proofEnd}

\Cref{lem:ILP} shows that it suffices for our algorithm to cut every Steiner tree enough times with high probability in order to yield a solution to the \ProblemName{Requirement Cut} problem.

Note that even if the $LPs$ are solved on modified graphs with added 0-cost edges, the final cut is sampled on the original graph $G$ and all terminal sets are cut in $r$ components of $G$. \Cref{lem:ILP} guarantees that cutting all Steiner trees in the modified graphs implies \textbf{satisfying the requirements in $G$}, therefore the validity and the cost of the solution remains intact and the number of Steiner trees remains the same since we can, by default, add all the 0 cost-edges to our cut if need be.

\subsection{Algorithm}
In this section we will see how our algorithm can give us a $\bigO(\log(\steinernumber{G,\Sets}))$-approximation for the \ProblemName{Requirement Cut} problem where $\steinernumber{G,\Sets}$ is the number of distinct minimal Steiner trees in the graph with respect to the given terminal sets. Computing $\steinernumber{G,\Sets}$ is $\texttt{\#}P$-Complete since it is a more general case of counting paths which is $\texttt{\#}P$-Complete \cite{valiant_complexity_1979}. However, we can upper-bound it by $\tau(G)\cdot g$ (\Cref{clm:spann-stein}) which is easy to compute through Kirchhoff's Theorem \cite{kirchhoff_ueber_1847}. When clear from the context, we shall denote the number $\steinernumber{G,\Sets}$ simply by $\sigma$.

\begin{algorithm}
\caption{Randomized Rounding for the \ProblemName{Requirement Cut} problem}
\label{alg:RRRC}
\begin{algorithmic}[1]
\State \textbf{Input:} Instance of the \ProblemName{Requirement Cut} problem on graph $G = (V, E)$
\State Solve the LP relaxation to obtain fractional values $\{d(e) \in [0,1] \mid e \in E\}$
\State Initialize cut set $\C \gets \emptyset$ 
\State Choose $\alpha = \frac{1}{c \cdot \log(\sigma)}$ with $c$ large enough (see \nameref{par:alph})
\For{each edge $e \in E$}
    \State Draw $x_e \sim \text{Uniform}(0,\alpha)$
    \If{$x_e \leq d(e)$}
        \State Add $e$ to $C$
    \EndIf
\EndFor
\State \textbf{Output:} Cut set $C$
\end{algorithmic}
\end{algorithm}
The rest of this section of the paper will be dedicated to proving the following theorem:
\begin{theorem}
\label{thm:main}
Let $I$ be a \ProblemName{Requirement Cut} Instance (\Cref{def:RCI}) with $\sigma$ being the number of minimal Steiner trees in $I$. Then, there is an LP rounding algorithm for the \ProblemName{Requirement Cut} problem that returns a solution of cost $\bigO(\log(\sigma)) \cdot LP_{opt}$ with probability $1 - \frac{1}{poly(\sigma)}$, where $LP_{opt}$ is the value returned by \RLP{}.
\end{theorem}
The algorithm provides an $\bigO(\log(n))$-factor approximation for all graphs with $\tau(G) = \bigO(n^c)$ (see \Cref{sec:GraphsBST} for examples). It also matches the $\bigO(\log(g))$-approximation on trees given by \cite{gupta_improved_2010}, while having a simpler process, as it only involves one of their two phases.

\begin{theoremEnd}[normal]{lemma}[]
The expected cost of $C$, the set of edges returned by \Cref{alg:RRRC} is $\Expected(c(C)) \leq \frac{1}{\alpha}\sum_{e \in E} c_e d_e = \bigO(\log(\sigma))\cdot LP_{opt}$.  
\end{theoremEnd}

\begin{proofEnd} 
The probability of a given edge $e$ being in the cut $C$ is $\mathbb{P}(e \in C) = \min\left\{1, \frac{d(e)}{\alpha} \right\}$. Thus, the expected cost of $C$ is $\Expected(c(C)) \leq \frac{1}{\alpha}\sum_{e \in E} c_e d_e = \bigO(\log(\sigma))\cdot LP_{opt}$ by linearity of expectations.
\end{proofEnd}

\begin{theoremEnd}[normal]{lemma}[$\star$]
\label{lem:anygroupsat}
After the execution of \Cref{alg:RRRC} the probability that a group $S_i$ lies in less than $r_i$ is $\frac{1}{\sigma^{\frac{c}{2}-1}}$ with $c \geq 4$ a constant chosen arbitrarily. 
\end{theoremEnd}

\begin{proofEnd}
Let $\alpha = \frac{1}{c \cdot \log(\sigma)}$ with $c$ a constant chosen large enough (\nameref{par:alph}) and define $x_e \sim \text{Uniform}(0, \alpha)$ picked independently for each $e \in E$. Let $d: E \to [0,1]$ be the solution returned by \RLP{}, and let $\iTrees$ be the set of all Steiner trees of group $g_i$ in $G$. We call a terminal set $i$ that lies in at least $r_i$ components \textit{satisfied}. By LP feasibility, we have:
\begin{equation}
\label{eqa:stree}
\sum_{e \in t} d(e) \geq r_i - 1 , \quad \forall t \in \iTrees    
\end{equation}

We define an indicator variable $Z_e$ for each edge $e \in E$ as follows:
\[
Z_e =
\begin{cases} 
0, & \text{if } x_e > d_e  , \\
1, & \text{otherwise}.
\end{cases}
\]

Next, define a random variable $Y_t$ for any Steiner tree $t \in \iTrees$:
\[
Y_t = \sum_{e \in t} Z_e.
\]

\begin{claim}
\label{clm:enoughcut}
For all $t \in \iTrees$, we have $\Expected[Y_t] \geq r_i - 1$. In other words, the expected number of edges cut in each Steiner tree is at least $r_i - 1$.
\end{claim}

\begin{proof}
We rely on \cref{eqa:stree} and the probability that any edge $e$ is cut is given by
\[
\mathbb{P}(Z_e = 1) = \min\left\{1, \frac{d(e)}{\alpha} \right\}
\]
Taking expectations, we obtain:
\[
\Expected[Z_e] \geq d(e).
\]
Summing over all edges in $t$, it follows that
\[
\Expected[Y_t] = \sum_{e \in t} \Expected[Z_e] \geq \sum_{e \in t} d(e) \geq r_i - 1.
\]
This completes the proof of this Claim.
\end{proof}

Now that we have established that each Steiner tree is expected to be cut at least $r_i - 1$ times, we analyze the probability that every $S_i$-Steiner tree is cut $r_i - 1$ times. Proving that our algorithm returns a solution using \Cref{lem:ILP}.

For this analysis, we shall first do a few adjustments into our instance: 

\begin{enumerate}
    \item Separate into two sets the edges of our cut $\C = \C_1 \uplus \C_2$. On the one hand, $\C_1$ is the set of edges such that $d(e) \geq \alpha \iff \mathbb{P}(Z_e = 1) = 1$. $\C_2$, on the other hand, is the set of edges for which $d(e) < \alpha$, meaning they had a probability $\mathbb{P}(Z_e = 1) = \frac{d(e)}{\alpha}$ of being picked in the cut.
    \item For each set $S_i$, after the removal of the edges of $\C_1$, some of its requirement might be satisfied. So we call $r_i'$ the residual requirement after the removal of the edges in $\C_1$. We have that $r_i' \leq r_i$. 
\end{enumerate}

The goal of the next steps of the analysis will be to show that the edges in $\C_2$ cover the residual requirements $r_i'$ left after the removal of the edges in $\C_1$ with high probability.

Since we are now only looking at the new graph $G = (V, E\setminus C_1)$, let us update some of our bounds from before. 

\[\forall e \in E\setminus C_1, \mathbb{P}(Z_e = 1) = \Expected[Z_e] = \frac{d(e)}{\alpha}\]

By linearity of expectations we get : $\forall i, \forall t \in \iTrees, \Expected[Y_t] = \sum_{e \in t} \Expected[Z_e] \geq \sum_{e \in t} \frac{d(e)}{\alpha} \geq \frac{r'_i - 1}{\alpha}$. Now, let us bound the probability of one Steiner tree not being satisfied using a Chernoff bound. We shall drop the indices on the variables since the calculations hold for every Steiner tree $t$ in $G = (V, E\setminus C_1)$.

We use the following Chernoff bound:

\begin{equation*}
\
\mathbb{P}\left(Y_t < (1 - \delta) \Expected[Y_t]\right) \leq e^{-\frac{\delta^2}{2} \Expected[Y_t]}, \quad \text{for } 0 < \delta < 1.
\label{eq:chernoff}
\end{equation*}

Substituting $\delta = 1 - \frac{r'-2}{\Expected[Y_t]}$, we derive:
\[
\mathbb{P}(Y_t \leq r'-2) \leq e^{-\frac{(1 - \frac{r'-2}{\Expected[Y_t]})  ^2}{2} \Expected[Y_t]}.
\]
Expanding the exponent:
\[
1 - \frac{r'-2}{\Expected[Y_t]} = \frac{\Expected[Y_t] - r' - 2}{\Expected[Y_t]},
\]
\[
(1 - \frac{r'-2}{\Expected[Y_t]})^2 = \frac{(\Expected[Y_t] - r' - 2)^2}{\Expected[Y_t]^2}.
\]
Substituting into the exponent:
\[
\exp({-\frac{(\Expected[Y_t] - r' - 2)^2}{2 \Expected[Y_t]}})
\]

Using the lower bound on $\Expected[Y_t]$:
\[
\Expected[Y_t] \geq \frac{r' - 1}{\alpha}
\]

Using this bound to upper bound the original expression:

Note: This upper-bounding is possible because the function $f(x) = \exp\left(-\frac{(x - r' - 2)^2}{2 x}\right)$ is decreasing for any $x \geq r'-2$ which is the case for our case.
\[
\exp\left(-\frac{(\Expected[Y_t] - r' - 2)^2}{2 \Expected[Y_t]}\right)
\leq
\exp\left(-\frac{\left(\frac{r' - 1}{\alpha} - r' - 2\right)^2}{2 \cdot \frac{r' - 1}{\alpha}}\right)
\]

Simplifying the numerator:
\[
\frac{r' - 1}{\alpha} - r' - 2
= \frac{r' - 1 - \alpha(r' + 2)}{\alpha}
\]

Squaring the numerator:
\[
\left(\frac{r' - 1 - \alpha(r' + 2)}{\alpha}\right)^2
= \frac{(r' - 1 - \alpha(r' + 2))^2}{\alpha^2}
\]

Putting this into the full expression:
\[
\exp\left(-\frac{(r' - 1 - \alpha(r' + 2))^2}{2 \cdot \frac{r' - 1}{\alpha} \cdot \alpha^2}\right)
= \exp\left(- \frac{(r' - 1 - \alpha(r' + 2))^2}{2 \alpha (r' - 1)}\right)
\]

Finally, we know that $r' \geq 2$ and we want to substitute. However, we have to make sure that this correctly bounds our term. Similarly to before this upper bounding is indeed possible because $f_{\alpha}(x) = \exp\left(- \frac{(x - 1 - \alpha(x + 2))^2}{2 \alpha(x - 1)}\right)$ admits only one critical points for $x \geq 2$  and it is $x^* = \frac{1 + 2\alpha}{1 - \alpha}$. The function increases before that and decreases after. We want it to be decreasing for $x \geq 2$ and so we have to restrict $\alpha$ to $0<\alpha\leq 0.25$ and $x \geq 2$. This is no problem since we can decide on $c$ and make $\alpha$ small enough.
\[
\exp\left(- \frac{(r' - 1 - \alpha(r' + 2))^2}{2 \alpha (r' - 1)}\right) \leq \exp\left(- \frac{(1 - 4\alpha)^2}{2 \alpha}\right)
\]
\[
\leq \exp\left(- \frac{16\alpha^2 - 8\alpha + 1}{2 \alpha}\right)
\]

Substitute $\alpha$ for $\frac{1}{c\cdot \log(\sigma)}$:

\[
 \exp\left(- \frac{\frac{16}{(c\cdot \log(\sigma))^2} - \frac{8}{c\cdot \log(\sigma)} + 1}{\frac{2}{c\cdot \log(\sigma)}}\right) =  \exp\left(- \frac{\frac{16}{(c\cdot \log(\sigma))} - 8 + c\cdot \log(\sigma)}{2}\right)
\]
Here we simply distribute the division by two and then remove a negative term.
\[
 =  \exp\left( 4 - \frac{8}{c\cdot \log(\sigma)} - \frac{c\cdot \log(\sigma)}{2}\right) \leq  \exp\left( 4 - \frac{c\cdot \log(\sigma)}{2}\right)
\]

\[
 =  \frac{\exp(4)}{\sigma^{\frac{c}{2}}}
\]

Then the probability of $C$, the output of the algorithm, \textbf{not} being a solution is at most:
\[
\sum_{S_i \in \Sets}\sum_{e \in \iTrees} \frac{\exp(4)}{\sigma^{\frac{c}{2}}} = \frac{\exp(4)\cdot\sigma}{\sigma^{\frac{c}{2}}} = \frac{\exp(4)}{\sigma^{\frac{c}{2}-1}}
\]
Under the condition that $\sigma = poly(n)$, this means that $\C{}$ is a solution with high probability i.e. $1-\frac{1}{poly(n)}$.

This concludes the proof of \Cref{lem:anygroupsat}
\end{proofEnd}

This also completes the proof of \Cref{thm:main}.

\paragraph*{\textbf{Choosing alpha}} \label{par:alph}
Let us now discuss the details of fixing $\alpha = \frac{1}{c \cdot \log(\sigma)}$. For some part of our analysis (see long version of proof of \Cref{lem:anygroupsat}) we need $c \geq 4$ assuming $\log(n) \geq 1$.
Technically speaking, one could have a $\log(\sigma)$-approximation even for instances where $\sigma$ is not $poly(n)$. But this would not lead to a ``single-log'' approximation, as the factor is only $\bigO(\log(n))$ when $\sigma = poly(n)$ or smaller.

\section{Depth of Series Parallel Graphs}
\label{sec:depthspg}
In this section we refine the approximation bounds for the \ProblemName{Requirement Cut} problem on series-parallel graphs by reanalyzing the embedding algorithm of Emek and Peleg~\cite{emek_tight_2010}. Combined with the $\bigO(\log g)$-approximation for trees~\cite{gupta_improved_2010}, this yields an $\bigO(m \log g)$-approximation on series-parallel graphs of depth $m$ (as in \Cref{def:depth}). 

Series-parallel graphs in general cannot be embedded into trees with expected distortion $o(\log n)$~\cite{gupta_cuts_2004}. Hence, for unrestricted series-parallel graphs, the best known ratio for \ProblemName{Requirement Cut} matches the general case, namely $\bigO(\log g \cdot \log n)$ via tree embeddings~\cite{nagarajan_approximation_2010,gupta_improved_2010}. The lower-bound construction of Gupta et al.~\cite{gupta_cuts_2004} in fact has depth $\Theta(\log n)$, so our result gives a strictly better guarantee whenever the depth $m$ is $o(\log n)$ (and in particular when $m$ is constant).
\begin{figure}[ht]
\small
	\begin{center}
		\begin{tikzpicture}
	   \node at (0,0){\includegraphics[scale = 0.8]{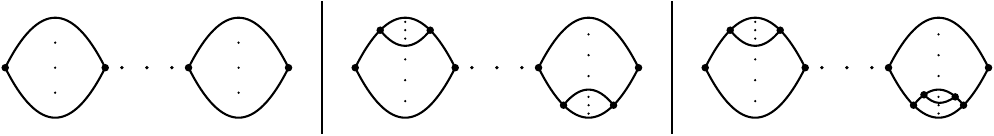}};
        \node at (-4.65,-1.1){depth $ = 3$};
        \node at (0,-1.1){depth $ = 5$};
        \node at (4.65,-1.1){depth $ = 7$};
		\end{tikzpicture}  
	\end{center}
	\caption{Series-parallel graphs of \depth 3, 5, and 7. These depths correspond to the number of nested compositions in the graph's construction and form an infinite subclass of graphs with bounded tree-depth. Note that the number of paths in those graphs can be exponential. Here, each black line represents an arbitrarily long path.}
	\label{fig:boundeddepthsp}
\end{figure}
Additionally, series-parallel graphs of bounded depth are a generalization of another well known class of structures in physics known as \textit{melonic graphs} \cite{aluffi_motives_2022, baratin_melonic_2014}. Those graphs are used in modeling condensed matter and black holes through the Sachdev–Ye–Kitaev Model (SYK) and its generalization the Carrozza, Tanasa, Klebanov, and Tarnopolsky (CTKT) model. \textit{Melonic graphs} are built recursively using an operation called ``bananification'' (\cite{aluffi_motives_2022} did not invent it but uses it), this operation creates series-parallel graphs with some \depth that corresponds to the order of the \textit{melonic graph}.

The depth parameter that we consider captures exactly the depth of the tree decomposition of the given series-parallel graph. This draws a parallel to the well-studied structural parameter tree-depth~\cite{nesetril_bounded_2012}, which can be equivalently seen to capture the width of the tree decomposition together with its depth. As series-parallel graphs are of treewidth 2, series-parallel graphs of bounded depth are thus a special case of graphs of bounded tree-depth.

In 2006, Emek and Peleg \cite{emek_tight_2010} presented an algorithm that probabilistically embeds series-parallel graphs into distributions over their spanning trees with expected distortion $\bigO(\log(n))$. Their approach constructs a spanning tree recursively from the \hyperref[def:comptrace]{composition trace} of the graph, ensuring that, in expectation, distances between any pair of vertices remain within a logarithmic factor of the original graph.

We modify the analysis by focusing on the depth of the series-parallel graph. Note that the depth of a series-parallel graph can grow linearly in the worst case (it is upper bounded by $\frac{n}{2}$). However, our claim is not to have a better approximation for all series-parallel graphs but simply that for any fixed depth, all series-parallel graph of that depth have an efficient approximation. As discussed in the introduction, series-parallel graphs of depth $\bigO(\log n)$ pose a serious obstacle towards better approximation for \ProblemName{Requirement Cut} since they are as hard as general graphs to approximate.

In other words, our analysis proves that fixing the depth at a constant $m$ for series-parallel allows $\bigO(\log(g))$ for the \ProblemName{Requirement Cut} problem on this class of graphs.

Before giving our theorem and proving it, it is good to give some context regarding Emek and Peleg's algorithm  \textit{Construct\_Tree}. It takes as input a series-parallel graph $G$ with two terminal vertices $x$ and $y$ and returns a spanning tree $T$ of $G$. The goal is to minimize the expected distortion between the shortest distances in $G$ and in $T$.

The algorithm operates recursively, using the composition tree (or decomposition trace) of $G$, which reflects how the graph was built from smaller series-parallel components via \emph{series} and \emph{parallel} composition. It handles these two cases differently:

\begin{itemize}
  \item \textbf{Series composition:} If $G$ is obtained by gluing $G_1, \ldots, G_k$ in series, the algorithm recursively constructs a spanning tree $T_j$ for each $G_j$, then glues $T_1, \ldots, T_k$ in series to form $T$. This does not increase distortion for any pair $(u,v)$ contained in any $G_i$.

  \item \textbf{Parallel composition:} If $G$ is obtained by taking $G_1, \ldots, G_k$ that each connect $x$ to $y$, the algorithm first builds a spanning tree $T_j$ for each $G_j$. It identifies a shortest $(x,y)$-path $P_{k^*}$ among them, and keeps $T_{k^*}$ untouched. For the other $T_j$ (with $j \neq k^*$), it removes one edge chosen uniformly at random from the path $P_j$ connecting $x$ and  $y$ in $T_j$. This breaks the cycles and ensures the union of all modified trees still forms a spanning tree of $G$. This step might increase the distortion for some $(u,v)$ contained in any $G_i$. 
\end{itemize}

The number of compositions is bounded by $m$ since it is the maximum number of compositions we do. To measure the distortion, Emek and Peleg introduce several variables: $\{a_i\}_{i=1}^m$, $\{b_i\}_{i=1}^m$ and $d$ positive integers satisfying:

\noindent
\begin{minipage}[t]{0.48\textwidth}
\vspace{0pt}
    \begin{itemize}
    \item $b_1 = 1$,
    \item $b_{i+1} \leq a_i + b_i$ for all $1 \leq i < m$,
    \item $c_i = \sum_{j=1}^{i-1} a_i$,
\end{itemize}
\end{minipage}%
\hfill
\begin{minipage}[t]{0.48\textwidth}
\vspace{0pt}
    \begin{itemize}
    \item $d \leq a_m + b_m$,
    \item $\sum_{i=1}^m a_i \leq n$.
\end{itemize}
\end{minipage}

Intuitively, these variables each mean something different in the construction with respect to a pair of vertices $u,v \in G_j$, a graph that does not contain $P_k$ the shortest $(x,y)$-path. During a parallel composition operation, $d(u,v)$ might increase if an edge in the shortest $(u,v)$-path is removed:
\begin{itemize}
    \item $b_i$ is the length of the overlap between the shortest $(u,v)$-path and the $(x,y)$-path in $G_j$, $P_j$. This is the size of the ''risk`` portion of $P_j$.
    \item $a_i$ is the length of the rest of $P_j$.
    \item $c_i$ is $d(u,v) - b_i$, it is what will remain of the $d(u,v)$ no matter what.
    \item $d$ is simply $d(x,y)$.
\end{itemize}

We will now show how to bound the expected distortion of Emek and Peleg's algorithm by an $\bigO(m)$ factor instead of a $\bigO(\log(n))$. 

\begin{theoremEnd}[normal]{theorem}[]
Let $\mathcal{T}$ be the output of Algorithm \textit{Construct\_Tree} \cite{emek_tight_2010} when invoked on a series-parallel graph G with $m$ the depth of the composition trace of $G$. Then $\Expected[dist_{\mathcal{T}}(u, v)] \leq \bigO(m), \forall (u, v) \in E(G)$. 
\end{theoremEnd}
\begin{proofEnd} 

Let $\Expected[dist_{\mathcal{T}}(u, v)]$, be the distortion between $u,v \in V$ as defined in \cite[Lemma 4.1]{emek_tight_2010} :
$$\Expected[dist_{\mathcal{T}}(u, v)] =
\sum\limits_{i=1}^{m} \frac{a_i}{a_i + b_i}
\left( \prod\limits_{j=1}^{i-1} \frac{b_j}{a_j + b_j} \right) (b_i + c_i)
+ \left( \prod\limits_{i=1}^{m} \frac{b_i}{a_i + b_i} \right) (a_m + c_m + d_m)$$

In their paper \cite{emek_tight_2010}, Emek and Peleg show that $\Expected[dist_{\mathcal{T}}(u, v)] \leq \bigO(\log n)$ (Lemma 4.2), meaning that the expected distortion of each edge is at most logarithmic in $n$. That is:

$$\Expected[dist_{\mathcal{T}}(u, v)] \leq \bigO(\log n) \quad \text{for all } (u,v) \in E(G)$$

We adopt the same notations as them and with a different analysis obtain a different result depending on $m$: $\Expected[dist_{\mathcal{T}}(u, v)] \leq \bigO(m), \forall (u, v)\in E(G)$

We first rewrite it using $b_{i+1} \leq b_i+a_i$ which implies that $b_i \leq 1 + \sum_{j=1}^{i-1} a_i$

$$ \Expected[dist_{\mathcal{T}}(u, v)] \leq
\frac{a_1}{a_1 + 1} + \sum\limits_{i=2}^{m} \frac{a_i}{1+\sum_{j=1}^{i} a_j} \left(\frac{1}{1 + \sum_{j=1}^{i-1} a_j} \right) (1 + \sum\limits_{j=1}^{i-1} a_j + c_i)
+ \left(\frac{(a_m + c_m + d_m)}{1 + \sum_{i=1}^{m} a_i} \right)
$$

We replace $c_i$, using $c_i = a_1 + ... + a_{i-1}$

$$ =
\frac{a_1}{a_1 + 1} + \sum\limits_{i=2}^{m} \frac{a_i}{1+\sum_{j=1}^{i} a_j} \left(\frac{1 + 2\cdot \sum_{j=1}^{i-1} a_j}{1 + \sum_{j=1}^{i-1} a_j} \right)
+ \left(\frac{d_m + \sum_{i=1}^{m} a_i}{1 + \sum_{i=1}^{m} a_i} \right)
$$

And now we upper bound $d_m$ by using $d_m \leq 1 + \sum_{i=1}^{m} a_i$

$$ \leq
\frac{a_1}{a_1 + 1} + \sum\limits_{i=2}^{m} \frac{a_i}{1+\sum_{j=1}^{i} a_j} \left(\frac{1 + 2\cdot \sum_{j=1}^{i-1} a_j}{1 + \sum_{j=1}^{i-1} a_j} \right)
+ \left(\frac{1 + 2 \cdot \sum_{i=1}^{m} a_i}{1 + \sum_{i=1}^{m} a_i} \right)
$$

The first sum has $m$ terms each smaller than 1 and each of them are multiplied by a quantity smaller than 2. The rightmost term is at most 2.

$$ \Expected[dist_{\mathcal{T}}(u, v)] \leq
2m + 2
$$

This series of calculations show that for a given \depth, all series-parallel graphs of that \emph{depth} can be embedded into a distribution of spanning trees that give a constant expected distortion for any pair of vertices $(u,v)$. This fact is interesting on its own theoretically and in our case, we can use it to get an $\bigO(\log(g))$-approximation for the \ProblemName{Requirement Cut} problem on series-parallel graphs of bounded \depth.
\end{proofEnd}

\section{Conclusion}
\label{sec:conclusion}
The \ProblemName{Requirement Cut} problem generalizes many classical cut problems, and has resisted single-log approximations on general graphs due to unavoidable lower bounds like the cost of tree embeddings. In this work, we identify two structural parameters—\emph{the number of (minimal) Steiner trees} and the \emph{depth of series-parallel graphs}—allowing us to break the ``double-log'' barrier.

Our main takeaway is that graphs with a polynomial number of spanning trees, or more generally instances with a polynomial number of minimal Steiner trees, admit efficient LP-rounding algorithms with approximation ratio $\bigO(\log(n))$. Additionally series-parallel graphs of bounded depth can be approximated in $\bigO(\log(g))$, extending known results for trees to a broader class.

We believe our work represents a natural limit of improvement (up to constant factors) for both our methods. On one hand, $LP$ rounding techniques based on randomized thresholding have been extensively studied and our improvement implicitly defines and calls for graphs of polynomial number of spanning trees which have not been studied before. On the other hand our linear bound in the depth $m$ of series-parallel graphs is likely optimal up to constant factors, since the series-parallel graphs constructed in \cite{gupta_cuts_2004} which exhibit distortion $\Omega(\log n)$ when embedded into trees have depth $\Theta(\log n)$ matching the $\log(n)\cdot \log(g)$-approximation \cite{nagarajan_approximation_2010, gupta_improved_2010} on these graphs. Therefore, any future progress may require fundamentally different approaches. 

Nonetheless, our results suggest that graphs with polynomially many spanning trees and bounded-depth constructions behave in a ``tree-like'' manner and allow for single-log approximations for the \ProblemName{Requirement Cut} problem, opening the door to improved approximation algorithms for other connectivity problems in such graphs. 
\newpage

\bibliographystyle{splncs_srt}

\bibliography{references}

\end{document}